\documentclass[12pt, journal, draftclsnofoot, onecolumn]{IEEEtran}

\usepackage{cite}
\usepackage{mdwmath}
\usepackage[english]{babel}
\usepackage{epsfig}
\usepackage{easymat}
\usepackage{amsmath}
\usepackage{float}
\usepackage{cite}
\usepackage{graphicx}
\usepackage{balance}
\usepackage{subfigure}
\newtheorem{theorem}{Theorem}
\newtheorem{corollary}{Corollary}
\newtheorem{Lemma}{Lemma}

\begin{document}
\renewcommand{\textfraction}{0}

\title{Capacity of The Discrete-Time Non-Coherent Memoryless Rayleigh Fading Channels at Low SNR}
\author{Z. Rezki and David Haccoun \\ \small{Department of Electrical
Engineering, \'Ecole Polytechnique de Montr\'eal,\\ Email:
\{zouheir.rezki,david.haccoun\}@polymtl.ca,} \\
\large{Fran\c{c}ois Gagnon}
\\
\small{Department of Electrical Engineering, \'Ecole de
technologie sup\'erieure,\\ Email: francois.gagnon@etsmtl.ca } }
\date{}
\maketitle \thispagestyle{empty}
\begin{abstract}
The capacity of a discrete-time memoryless channel, in which
successive symbols fade independently, and where the channel state
information (CSI) is neither available at the transmitter nor at
the receiver, is considered at low SNR. We derive a closed form
expression of the optimal capacity-achieving input distribution at
low signal-to-noise ratio (SNR) and give the exact capacity of a
non-coherent channel at low SNR. The derived relations allow to
better understanding the capacity of non-coherent channels at low
SNR and bring an analytical answer to the peculiar behavior of the
optimal input distribution observed in a previous work by Abou
Faycal, Trott and Shamai. Then, we compute the non-coherence
penalty and give a more precise characterization of the sub-linear
term in SNR. Finally, in order to better understand how the
optimal input varies with SNR, upper and lower bounds on the
capacity-achieving input are given.
\end{abstract}
\normalsize

\section{INTRODUCTION}
In wireless communication, the channel estimation at the receiver
is not often possible due, for instance, to the high mobility of
the sender or the receiver or both. Therefore, achieving reliable
communication over fading channels where the channel state
information (CSI) is available neither at the transmitter nor at
the receiver, is of a particular interest. Establishing the
performance limits, in terms of channel capacity, error
probability, etc.., in such a non-coherent scenario has recently
motivated extensive works (see for example \cite{Verdu-02},
\cite{Medard-2000}). When CSI is available at the receiver, the
channel capacity, commonly known as the coherent capacity has been
studied by Ericson \cite{Ericson-70} for a Single Input Single
Output (SISO) channel and recently by many other authors for a
Multiple Input Multiple Output (MIMO) channel \cite{Foschini96}
\cite{ar-capacity}. Conversely, when CSI is not available at both
ends, computing the channel capacity, known as the non-coherent
capacity, as well as computing the optimal input distribution
achieving this capacity, for both SISO and MIMO channels, is a
rather tedious task \cite{Abou-Faycal-01} \cite{perera-06}. The
main difficulty in computing the non-coherent capacity relies on
the fact that the capacity-achieving input distribution is
discrete with a finite number of mass points, where one of them is
located at the origin. The number of these mass points increases
with the signal-to-noise ratio (SNR). Since no bound on the number
of mass points with respect to SNR is actually available, it is
very difficult to find closed form expressions for both the
achievable capacity and the optimal input distribution for all SNR
values. Fortunately, numerical computation of the capacity and the
optimal input distribution has been made possible using the
Khun-Tucker condition which is a necessary and sufficient
condition for optimality, for of a SISO channel
\cite{Abou-Faycal-01} and for a MIMO channel \cite{perera-06}.\

Earlier in 1999, using a block fading channel, Marzetta and
Hochwald have obtained the structure of the optimal input, with
explicit calculations for the special case of a SISO channel at
high SNR values or with a large coherence time \cite{Marzetta}.
The non-coherent capacity was also computed as a function of the
number of transmit and receive antennas as well as the coherence
time at high SNR in \cite{Zheng-2002}. At a low SNR regime, it was
also shown in \cite{Zheng-2002} that to a first order of magnitude
of the SNR, there is no capacity penalty for not knowing the
channel at the receiver which is not the case at the high SNR
regime. It has been well established previously that at low SNR,
just like in an additive white Gaussian noise (AWGN) channel, the
capacity of a fading channel varies linearly with the SNR
regardless of whether or not the CSI is available at the receiver
\cite{Kennedy-book}, \cite{Telatar2000}. Recently, this power
efficiency at a low SNR regime or equivalently at a large channel
bandwidth has motivated work towards a better understanding of the
non-coherent capacity at a low SNR regime \cite{Verdu-02},
\cite{Zheng-2007}, \cite{Ray-2007} for both SISO and MIMO channels
using several fading models.\

In this paper, we analyze the capacity of a discrete time
non-coherent memoryless Rayleigh fading SISO channel at low SNR.
The main contributions of this paper are:
\begin{enumerate}
    \item Derivation of an analytical closed form of the channel
    mutual information at low SNR, which may also be considered as a
    lower bound on the channel mutual information for an arbitrary
    SNR value.
    \item Derivation of a fundamental relation between the capacity-achieving input distribution
    and the SNR value, from which an exact capacity expression is deduced at low SNR.
    \item Derivation of novel upper and lower bounds on the non-zero mass point
    location of the optimal input, which allow to deduce lower and upper
    bounds respectively on the non-coherent capacity at low SNR.
\end{enumerate}

The paper is organized as follows. Section \ref{model} presents
the system model. In section \ref{mut-info}, we derive a closed
form expression of the channel mutual information at low SNR which
is also a lower bound on the channel mutual information at all SNR
values. The optimal input distribution as well as the non-coherent
capacity are presented in Section \ref{cap}. Numerical results are
reported in Section \ref{result} and Section \ref{conclusion}
concludes the paper.

\section{CHANNEL MODEL}\label{model}

We consider a discrete-time memoryless Rayleigh-fading channel
given by:
\begin{equation}\label{E1}
    r(l) = h(l) s(l) + w(l), \hspace{1cm} l=1,2,3,...
\end{equation}
where $l$ is the discrete-time index, $s(l)$ is the channel input,
$r(l)$ is the channel output, $h(l)$ is the fading coefficient and
$w(l)$ is an additive noise. More specifically, $h(l)$ and $w(l)$
are independent complex circular Gaussian random variables with
mean zero and variances $\sigma_{h}^{2}$ and $\sigma_{w}^{2}$,
respectively. The input $s(l)$ is subject to an average power
constraint, that is $E[|s(l)|^{2}] \leq P$, where $E[.]$ indicates
the expected value. It is assumed that the channel state
information is available neither at the transmitter nor at the
receiver. However, even though the exact values of $h(l)$ and
$w(l)$ are not known, their statistics are, at both ends.

Model (\ref{E1}) appears for example during the decomposition of a
wideband channel into parallel noninteracting channels, or when a
narrow-band signal is hopped rapidly over a large set of
frequencies, one symbol per hop \cite{Verdu-02}.

Since the channel defined in (\ref{E1}) is stationary and
memoryless, the capacity achieving statistics of the input $s(l)$
are also memoryless, independent and identically distributed
(i.i.d). Therefore, for simplicity we may drop the time index $l$
in (\ref{E1}). Consequently, the distribution of the channel
output $r$ conditioned on the input $s$ can be obtained after
averaging out the random fading coefficient $h$, yielding:
\begin{equation}\label{E2}
f_{r|s}(r|s)=\frac{1}{\pi(\sigma_{h}^{2}|s|^{2}+\sigma_{w}^{2})}\exp\left[\frac{-|r|^2}{\sigma_{h}^{2}|s|^{
2}+\sigma_{w}^{2}}\right].
\end{equation}
Noting that in (\ref{E2}), the conditional output distribution
depends only on the squared magnitudes $|s|^{2}$ and $|r|^{2}$, we
will no longer be concerned with complex quantities but only with
their squared magnitudes. Conditioned on the input, $|r|^{2}$ is
chi-square distributed with two degrees of freedom:
\begin{equation}\label{E3}
f_{|r|^{2}|s}(t|s)=\frac{1}{(\sigma_{h}^{2}|s|^{2}+\sigma_{w}^{2})}\exp\left
[\frac{-t}{\sigma_{h}^{2}|s|^{2}+\sigma_{w}^{2}}\right].
\end{equation}
Normalizing to unit variance, let $y=|r|^{2}/\sigma_{w}^{2}$ and
let $x=|s|\sigma_{h}\sigma_{w}$. Then (\ref{E3}) may be written
more conveniently as:
\begin{equation}\label{E4}
    f_{y|x}(y|x)=\frac{1}{(1+x^2}\exp\left[\frac{-y}{1+x^2}\right],
\end{equation}
with the average power constraint $E[x^{2}] \leq a$, where
$a=P\sigma_{h}^{2}/\sigma_{w}^{2}$ is the SNR per symbol time.

\section{THE CHANNEL MUTUAL INFORMATION}\label{mut-info}

For the channel (\ref{E4}), the mutual information is given by
\cite{Gallager-book}:
\begin{equation}\label{E5}
    I(x;y)=\int \int
f_{y|x}(y|x)f_{x}(x)\ln{\frac{f_{y|x}(y|x)}{f_{(y;x)}(y;x)}} dx
    dy.
\end{equation}
The capacity of channel (\ref{E4}) is the supremum
\begin{equation}\label{E6}
    C=\sup_{E[x^{2}] \leq a}I(x;y)
\end{equation}
over all input distributions that meet the constraint power. The
existence and uniqueness of such an input distribution was
established in \cite{Abou-Faycal-01}. More specifically, the
optimal input distribution for channel (\ref{E4}) is discrete with
a finite number of mass points, where one of them is necessarily
null. That is, the capacity (\ref{E6}) is expressed by
\begin{equation}\label{E7}
    C=\max_{E[x^{2}] \leq a}\sum_{i=0}^{N-1} p_{i} \int_{0}^{\infty}
    f_{y|x_{i}}(y|x_{i}) \ln \left[\frac{f_{y|x_{i}}(y|x_{i})}{\sum
    _{j}p_{j}f_{y|x_{j}}(y|x_{j})}\right]dy,
\end{equation}
where $x_{0}=0<x_{1}<x_{2} \ldots <x_{N-1}$ are the mass point
locations and where $p_{0},p_{1} \ldots ,p_{N-1}$ their
probabilities respectively. This optimization problem is very
difficult since the number of discrete mass points, the optimum
probabilities and their locations are unknown. In
\cite{Abou-Faycal-01}, numerical evaluation of the capacity and
the optimum input distribution was given using the Khun-Tucker
condition which is necessary and sufficient for optimality. The
authors have found empirically that two mass points are optimal
for low SNR and that the number of mass points increases
monotonically with SNR. Many other papers have used these results
in order to further understand the non-coherent capacity and the
optimal input distribution behavior as the SNR approaches zero
\cite{Zheng-2007},\cite{Ray-2007}.

Since we focus on the low SNR regime, we may use in (\ref{E7}) a
discrete input distribution with two mass points, where one of
them is null, to obtain the optimal capacity at low SNR.
Furthermore, this on-off signaling also provides a lower bound on
the non-coherent capacity for all SNR values. Clearly, using
computer simulation, it was shown  in \cite{Abou-Faycal-01} that
on-off signaling provides a tight lower bound on the capacity for
the SNR values considered. That is, a lower bound on the capacity
may be expressed by:
\begin{equation}\label{E8}
    C_{LB}=\max_{E[x^{2}] \leq a}I_{LB}(x;y),
\end{equation}
where $I_{LB}(x;y)$ is a lower bound on the channel mutual
information $I(x;y)$ given by:
\begin{equation}\label{E9}
   I_{LB}(x;y)=I_{LB}(x_{1},p_{1})= \sum_{i=0}^{1} p_{i}
\int_{0}^{\infty}
    f_{y|x_{i}}(y|x_{i}) \ln \left[\frac{f_{y|x_{i}}(y|x_{i})}{\sum
    _{j}p_{j}f_{y|x_{j}}(y|x_{j})}\right]dy,
\end{equation}
and the average constraint power becomes: $p_{1}x_{1}^{2} \leq a$.
Note that the optimization problem in (\ref{E8}) is less complex
than in (\ref{E7}) since we deal with only two unknowns $p_{1}$ ND
$x_{1}$. Furthermore, it is proven below that further
simplifications can be obtained, using the fact that
$I_{LB}(x_{1},p_{1})$ is monotonically increasing in $x_{1}$ and
thus the problem at hand may be reduced to a simpler maximization
problem without constraint. We summarize this
result in lemma \ref{L1}.\\

\begin{Lemma}\label{L1}
The optimal capacity at low SNR and a lower bound on it for all
SNR values is given by:
\begin{equation}\label{E10}
    C_{LB}=\max_{x_{1} \geq \sqrt{a}} I_{LB}(x_{1},a),
\end{equation}
where $I_{LB}(x_{1},a)$ is the channel mutual information for a
given mass point location $x_{1}$ and a given SNR value $a$.
Furthermore, $I_{LB}(x_{1},a)$ may be written as:
\begin{equation}\label{E11}
    I_{LB}(x_{1},a)= \begin{cases} a-a \left[
    \frac{\ln{(1+x_{1}^{2})}}{x_{1}^{2}}+\frac{1}{1+x_{1}^{2}}
    +\frac{x_{1}^{2}}{1+x_{1}^{2}}\cdot
{}_{1}F_{2}\left(1,\frac{1}{x_{1}^{2}},1+\frac{1}{x_{1}^{2}},-\frac{(1+x_{1}
^{2})(x_{1}^{2}-a)}{a} \right)
    \right]  \\

-\ln{\left(1-\frac{a}{x_{1}^{2}}\right)}-\ln{\left(1+\frac{a}{(1+x_{1}^{2})(
x_{1}^{2}-a)}\right)} & \text{if $x_{1} > \sqrt{a}$},\\
0& \text{if $x_{1}= \sqrt{a}$}
\end{cases}
\end{equation}
where $_{2}F_{1} (\cdot,\cdot,\cdot,\cdot )$ is the Gauss
hypergeometric function.
\end{Lemma}
\begin{proof}
For convenience, the proof is presented in Appendix \ref{app-1}.
\end{proof}\

In Lemma \ref{L1}, the existence of a maximum for a given SNR
value $a$ is guaranteed by the continuity of $I_{LB}(x_{1},a)$ and
the fact that it is bounded with respect to $x_{1}$ over the
interval $[\sqrt{a},\infty[$. This can be readily seen in Fig.
\ref{Fig1} where we have plotted the lower bound $I_{LB}(x_{1},a)$
for different values of $a$. As can be seen in Fig. \ref{Fig1},
$I_{LB}(x_{1},a)$ has a maximum for the 3 SNR regimes. The
existence of such a maximum is also rigorously established in
Appendix \ref{app-1}. Clearly, as was discussed in Appendix
\ref{app-1}, the maximization (\ref{E10}) is reduced to solving
the equation $\frac{\partial}{\partial x_{1}}I_{LB}(x_{1},a)$ for
a given SNR value $a$. Ideally, an analytical solution would
provide an insight as to how the non-coherent capacity and the
optimal input distribution vary with the SNR. However, solving
such an equation for arbitrary SNR values is very ambitious since
it involves an analytical solution to a transcendental equations.
Nevertheless, it is of interest to focus on the low SNR regime to
get the benefit of some advantageous simplifications in order to
elucidate the non-coherent capacity behavior at low SNR.

\section{NON-COHERENT CAPACITY AT LOW SNR }\label{cap}

In this section, we will use Lemma \ref{L1} to derive a
fundamental analytical relation between the optimal input
distribution at a low SNR regime and the particular SNR value $a$.
We show in Theorem \ref{T1} that this fundamental relation holds
up to an order of $a$ strictly less than 2. As is shown below, the
derived relation is very useful since it allows computing the
optimal input distribution for a given SNR value $a$ while
providing a rigorous characterization as to  how the non zero mass
point locations and their probabilities vary with $a$. Moreover,
the derived relation may be used to compute the exact non-coherent
capacity at low SNR values.
\subsection{A fundamental relation between the optimal input
distribution and the SNR}\label{subsectionA}

We present the fundamental relation between the optimal input distribution and the SNR value in the following Theorem:\\
\begin{theorem}\label{T1}
At a low SNR value $a$, the optimal input probability distribution
for an order of magnitude of $a$ strictly less than 2, is given
by:
\begin{equation}\label{E12}
    f_{x}(x)= \begin{cases}
x_{1}& \text{with probability $p_{1}=\frac{a}{x_{1}^2}$},\\
0& \text{with probability $p_{0}=1-p_{1}$},
\end{cases}
\end{equation}
where $x_{1}$ is the solution of the equation:
\begin{equation}\label{E13}
    x_{1}^{2}-(1+x_{1}^{2})\ln(1+x_{1}^{2})-\pi
\left(\frac{a}{x_{1}^{2}+x_{1}^{4}}\right)^{\frac{1}{x_{1}^{2}}}\csc{\left(\frac{\pi}{x_{1}^{2}}\right)}\left[1+x_{1}^2-\pi\cot{\left(\frac{\pi}{x_{1}^{
2}}\right)}+\ln{\left(\frac{a}{x_{1}^{2}+x_{1}^{4}}
    \right)}
    \right]=0.
\end{equation}
Furthermore, the non-coherent channel capacity is given by:
\begin{equation}\label{E14}
    C(a,x_{1})=a-a \cdot
    \frac{\ln{(1+x_{1}^{2})}}{x_{1}^{2}}-a^{1+\frac{1}{x_{1}^{2}}} \cdot
\frac{\pi
\csc{\left(\frac{\pi}{x_{1}^{2}}\right)}\left(\frac{1}{x_{1}^{2}+x_{1}^{4}}\right)^{\frac{1}{x_{1}^{2}}}}{1+x_{1}^{2}}
\end{equation}
\end{theorem}
\begin{proof}
For convenience, the proof is presented in Appendix \ref{app-2}.
\end{proof}\

Clearly, (\ref{E13}) is also a transcendental equation, for which
determining an analytical solution is a very tedious task.
Although it is very involved to derive an analytical solution of
(\ref{E13}) in the form of $x_{1}=f(a)$, it is of interest from an
engineering point of view, to resolve (\ref{E13}) numerically and
obtain the optimal $x_{1}$ for a given SNR value $a$. One may then
get the value of the non-coherent capacity by replacing in
(\ref{E14}) the obtained value of $x_{1}$. Moreover, (\ref{E13})
provides some insight on the behavior of $x_{1}$ as $a$ tends
toward zero. For example, using (\ref{E13}), one may determine the
limit of $x_{1}$ as $a$ tends toward zero. To see this, let $M$ be
this limit and let us assume that $M$ is finite. From Appendix
\ref{app-2}, we know that for the optimal input distribution, the
non-zero mass point location $x_{1}$ is greater than one. Thus,
its limit as $a$ tends toward zero is greater or equal than one $M
\geq 1$. Then, taking the limits on both sides of (\ref{E13}) as
$a$ goes to zero yields:
\begin{equation}\label{Elimitx1}
    M^2-(1+M^2)\ln{(1+M^2)}=0.
\end{equation}
That is, if $M$ is finite, it would be equal to zero, the unique
solution to (\ref{Elimitx1}), but this is impossible since $M \geq
1$. Hence, consistently with \cite{Abou-Faycal-01,Zheng-2007},
$\underset{a \rightarrow 0}{\lim}x_{1}=\infty$ . Furthermore, we
have found that (\ref{E13}) may be written in a more convenient
way as:
\begin{equation}\label{E15}
     a= \exp{\left[x_{1}^{2}W\bigl(k,\varphi{(x_{1})}\bigr)-x_{1}^{2}+\pi
\cot{\bigl(\frac{\pi}{x_{1}^{2}}\bigr)}+\ln{(x_{1}^{2})}+\ln{(1+x_{1}^{2})}-
1
    \right]},
\end{equation}
with $k=-1$ if $a \leq a_{0}$ and $k=0$ elsewhere, and where
$W{(\cdot,\cdot)}$ is the Lambert function, with $\varphi{(x)}$
given by:
\begin{equation}\label{E16}
\varphi{(x)}=-\frac{\sin{(\frac{\pi}{x^{2}})}(-x^{2}+\ln{(1+x^{2})}+x^{2}
\ln{(1+x^{2})}) }{\pi x^{2}} \cdot \exp{\left( \frac{-\pi
\cot{\left(\frac{\pi}{x^{2}}\right)}}{x^{2}} +1+\frac{1}{x^{2}}
\right)}.
\end{equation}
Also, $a_{0}$ is the solution of (\ref{E13}) for $x_{1}=x_{0}$,
where $x_{0}$ is the root of the equation
$\varphi{(x)}=-\frac{1}{e}$. The number $-\frac{1}{e}$ comes out
in our analysis from the fact that it is the unique point shared
by the principal branch of the Lambert function $W(0,x)$ and the
branch with $k=-1$, $W(-1,x)$. That is
$W(0,-\frac{1}{e})=W(-1,-\frac{1}{e})$. This guarantees the
continuity of $a$ in (\ref{E15}) for all $x_{1}$ values.
Numerically, we have found that $a_{0}=0.0582$ and
$x_{0}=\sqrt{3.93388}$. Hence, (\ref{E15}) may also be viewed as a
fundamental relation between the optimal input distribution and
$a$ for discrete-time non-coherent memoryless Rayleigh fading
channels at low SNR. On the other hand, (\ref{E15}) provides the
global answer as to how the non-zero mass point location of the
optimal on-off signaling and the SNR are linked together. For this
purpose, a simple analysis of (\ref{E15}) has been done and some
important results are recapitulated in the following
corollary.\\

\begin{corollary}\label{C1}
At low SNR, we have:
\begin{enumerate}
     \item For all $a \leq a_{0}$, $a_{0}=0.0582$, $a$ is an decreasing function
    with respect to $x_{1}$ and for all $a > a_{0}$, $a$ is an
    increasing function of $x_{1}$.
    \item For all $a$, $x_{1} \geq x_{0}$, where $x_{0}=\sqrt{3.93388}$.\label{C12}
    \item $\underset{x_{1}\rightarrow \infty}{\lim}{a}=0$.
\end{enumerate}
\end{corollary}\

Corollary \ref{C1} agrees with \cite{Abou-Faycal-01} where it was
shown using computer simulation that the non-zero mass point
location passes through a minimum before moving upward. However,
by specifying the edge point $(x_{0},a_{0})$, Corollary \ref{C1}
gives a more precise characterization concerning this peculiar
behavior of the non-zero mass point locations. Furthermore,
Corollary \ref{C1} also refines the lower bound on $x_{1}$, $x_{1}
> 1$ and derives $x_{0}$ as an improved lower
bound on the non-zero mass point location at low SNR. Moreover,
from (\ref{E15}), we may write:
\begin{equation}\label{E16prime}
    \ln{(a)}+x_{1}^2= x_{1}^{2}W\bigl(k,\varphi{(x_{1})}\bigr)+\pi
\cot{\bigl(\frac{\pi}{x_{1}^{2}}\bigr)}+\ln{(x_{1}^{2})}+\ln{(1+x_{1}^{2})}-
1.
\end{equation}
It is then easy to check that the right hand side (RHS) of
(\ref{E16prime}) is a decreasing function of $x_{1}$ for $x_{1} <
x_{0}$, which yields an upper bound on $x_{1}$:
\begin{equation}\label{E16second}
    x_{1}^2 \leq -\ln{(a)}+\xi_0,
\end{equation}
where $\xi_0=\ln{(a_0)}+x_{0}^2$, which is again consistent with
the upper bound derived in \cite{Zheng-2007}. Note that the upper
bound (\ref{E16second}) is valid for all $a \leq a_{0}$ whereas
the upper bound provided in \cite{Zheng-2007} holds for $a \ll
a_{0}$ for which $\xi_0$ is negligible. On the other hand,
combining (\ref{E16second}) and the lower bound on $x_{1}$
provided in Corollary \ref{C1} one may obtain:
\begin{equation}\label{E16ter}
  a^{\alpha}x_{0}^2  \leq a^{\alpha}x_{1}^2 \leq a^{\alpha}
  (\xi_0-\ln{(a)}).
\end{equation}
for all $\alpha > 0$. That is:
\begin{equation}\label{E16quat}
    \underset{a\rightarrow 0}{\lim}{\bigl(a^{\alpha}x_{1}^2\bigr)}=0,
\end{equation}
which means that $a^{\alpha}$ tends toward zero faster than
$x_{1}^2$ does toward infinity. This result may also be used to
gain further insight on the capacity behavior at low SNR. For
instance, from (\ref{E14}), we may write the non-coherent capacity
as:
\begin{equation}\label{E16fifth}
    C(a)=a+\textit{o}(a),
\end{equation}
where $\textit{o}(a)=-a \cdot
    \frac{\ln{(1+x_{1}^{2})}}{x_{1}^{2}}-a^{1+\frac{1}{x_{1}^{2}}} \cdot
\frac{\pi
\csc{\left(\frac{\pi}{x_{1}^{2}}\right)}\left(\frac{1}{x_{1}^{2}+x_{1}^{4}}\right)^{\frac{1}{x_{1}^{2}}}}{1+x_{1}^{2}}$,
meaning that the non-coherent capacity varies linearly with $a$ at
low SNR and hence non-coherent communication at low SNR may be
qualified as energy efficient communication.
\subsection{Energy efficiency and non-coherence penalty }
In general, the capacity of a channel including a Gaussian channel
and a Rayleigh channel varies linearly at low SNR
\cite{Zheng-2007}. The difference between these channels in terms
of capacity can only be explained by the sub-linear term
$\textit{o}(a)$ in (\ref{E16fifth}). The sub-linear term has been
defined in \cite{Zheng-2007} as:
\begin{equation}\label{E16-subdef}
    \Delta(a):=a-C(a).
\end{equation}
At low SNR, the sub-linear term $\Delta(a)$ is also related to the
energy-efficiency. let $E_{n}$ be the transmitted energy in Joules
per information nat, then we have:
\begin{equation}\label{E16-0}
    \frac{E_{n}}{\sigma_{w}^{2}} \cdot C(a)=a.
\end{equation}
Using (\ref{E16-subdef}), we can write:
\begin{equation}\label{E16-01}
    \frac{E_{n}}{\sigma_{w}^{2}} =\frac{1}{1-\frac{\Delta(a)}{a}}
    \approx 1+\frac{\Delta(a)}{a},
\end{equation}
where the approximation holds if $\frac{\Delta(a)}{a}$ is
sufficiently small. Note that if
\begin{equation}\label{E16-1}
    \frac{\Delta(a)}{a}\rightarrow 0,
\end{equation}
then from (\ref{E16-subdef}) and (\ref{E16-01}), we have
respectively:
\begin{eqnarray}
\label{E116_appC}
  C(a) &\approx& a \\
  \label{E116_appC00}
  \frac{E_{n}}{\sigma_{w}^{2}} &\approx& 1,
\end{eqnarray}
which implies that the highest energy efficiency of -1.59 (dB) per
information bit could be theoretically achieved. For a Gaussian
channel and a fading channel under the coherent assumption, the
sub-linear terms are respectively given by \cite{Zheng-2007}:
\begin{eqnarray}
\label{E116 appC1}
  \Delta_{AWGN}(a) &=& \frac{1}{2}a^2+o(a^2) \\
  \label{E116 appC2}
  \Delta_{coherent}(a) &=& \frac{1}{2}E[\|h\|^4]a^2+o(a^2)
\end{eqnarray}
For a non-coherent Rayleigh fading channel, the sub-linear term
can be computed using (\ref{E14}):
\begin{equation}\label{E16-2}
    \Delta(a)=a \cdot
    \frac{\ln{(1+x_{1}^{2})}}{x_{1}^{2}}+a^{1+\frac{1}{x_{1}^{2}}} \cdot
\frac{\pi
\csc{\left(\frac{\pi}{x_{1}^{2}}\right)}\left(\frac{1}{x_{1}^{2}+x_{1}^{4}}\right)^{\frac{1}{x_{1}^{2}}}}{1+x_{1}^{2}}.
\end{equation}
Note that at very low SNR and following (\ref{E16-2}),
$\frac{\Delta(a)}{a}$ converges to zero making the non-coherent
Rayleigh channel also energy efficient. However, as SNR increases,
the convergence of $\frac{\Delta(a)}{a}$ to zero is slower than
$\frac{\Delta_{AWGN}(a)}{a}$ and $\frac{\Delta_{coherent}(a)}{a}$.
This could be seen from (\ref{E16quat}) indicating that $x_{1}$
converges slower to infinity than $a$ does to zero. To illustrate
this, as an example, let us calculate the value of
$\frac{\Delta(a)}{a}$ for an SNR value $a=-30dB$. Following
(\ref{E16-2}), we can write:
\begin{equation}\label{E16-40}
    \frac{\Delta(a)}{a}=\frac{\ln{(1+x_{1}^{2})}}{x_{1}^{2}}+a^{\frac{1}{x_{1}^{2}}} \cdot
\frac{\pi
\csc{\left(\frac{\pi}{x_{1}^{2}}\right)}\left(\frac{1}{x_{1}^{2}+x_{1}^{4}}\right)^{\frac{1}{x_{1}^{2}}}}{1+x_{1}^{2}}.
\end{equation}
Solving (\ref{E15}) for $a=-30dB$ with respect to $x_{1}^2$
yields: $x_{1}^2 \approx 4.96815$. Then, substituting this value
in (\ref{E16-40}), we obtain $\frac{\Delta(a)}{a}\approx 49\%$.
Note that for AWGN and coherent Rayleigh fading channels,
$\frac{\Delta_{AWGN}(a)}{a}$ and $\frac{\Delta_{coherent}(a)}{a}$
are at the same order of magnitude than the SNR value in this
case. It takes a lower SNR for non-coherent communication to
achieve the same energy efficient as AWGN and coherent Rayleigh
fading channels.\

In the range of SNR values of interest, we may define the
non-coherence penalty per SNR as:
\begin{equation}\label{E16-3}
    \frac{C_{coherent}(a)-C(a)}{a}.
\end{equation}
where $C_{coherent}$ is the channel capacity under coherent
assumption. Now, from \cite{Zheng-2007}, we can write
$C_{coherent}$ as:
\begin{equation}\label{E16-10}
    C_{coherent}(a)=a+\textit{O}(a)=a+\textit{o}(a^{2-\alpha}),
\end{equation}
for any $1 > \alpha > 0$. Recalling that the non-coherent capacity
in (\ref{E14}) was obtained using series decomposition to an order
strictly smaller than 2, then combining (\ref{E14}) and
(\ref{E16-10}), we derive the exact non-coherence penalty per SNR
up to this order:
\begin{equation}\label{E16-4}
    \frac{C_{coherent}(a)-C(a)}{a}=
    \frac{C_{coherent}-C}{C_{coherent}}=
    \frac{\ln{(1+x_{1}^2)}}{x_{1}^2}+a^{\frac{1}{x_{1}^{2}}} \cdot
\frac{\pi
\csc{\left(\frac{\pi}{x_{1}^{2}}\right)}\left(\frac{1}{x_{1}^{2}+x_{1}^{4}}\right)^{\frac{1}{x_{1}^{2}}}}{1+x_{1}^{2}}
\end{equation}
Now using (\ref{E16quat}), dividing both sides of (\ref{E16-4}) by
$a^{\alpha}$, ($\alpha > 0$) and taking the limit as $a$ tends to
zero yields:
\begin{equation}\label{E16-6}
    C_{coherent}(a)-C(a) \gg a^{1+\alpha},
\end{equation}
where $\gg$ means:
\begin{equation}\label{E16-7}
    \underset{a \rightarrow
    0}{\lim}\frac{C_{coherent}(a)-C(a)}{a^{1+\alpha}}=\infty.
\end{equation}
Inequality (\ref{E16-6}) indicates that not only the non-coherent
capacity is much greater than $a^2$ as was established in
\cite{Verdu-02}, but more precisely, it is much greater than
$a^{1+\alpha}$ since $a^{1+\alpha} \gg a^{2}$, $1 > \alpha > 0$.
Again, this result is in full agreement with \cite{Zheng-2007}.\\

In this subsection, we have discussed exact closed forms of the
optimal input distribution and the non-coherent capacity based on
the fundamental relation (\ref{E13}) or equivalently (\ref{E15}).
However, one may be interested in deriving simpler lower and upper
bounds on these quantities in order to better understand how they
vary with the SNR value $a$. This is discussed next.\

\subsection{Upper and lower bounds on the non-coherent capacity  }
Considering (\ref{E15}), since we are interested in the low SNR
regime, we assume for simplicity that $a \leq a_{0}$. Thus the
Lambert function in (\ref{E15}) is the branch with $k=-1$, that is
$W(-1,x)$. A lower bound on the non-coherent capacity is easily
obtained by combining (\ref{E16second}) and (\ref{E14}) and will
be referred to as $C_{LB}(a)$. We now derive the lower bound on
the optimal non-zero mass point location and the upper bound on
the non-coherent capacity in Theorem \ref{T2}.\

\begin{theorem}\label{T2}
At low SNR values $a$, a lower bound on the optimal non-zero mass
point location is given by:
\begin{equation}\label{E17}
    x_{1,
    LB}=\frac{y}{\sqrt{-W\Biggl(-1,\varphi{\Bigl(\frac{y}{ -\ln{\bigl(
-\varphi{(y)}         \bigr)} }\Bigr)}\Biggr)}},
\end{equation}
where $y=\sqrt{1+\ln{\frac{1}{a}}}$. Furthermore, an upper bound
on the non-coherent capacity can be obtained from (\ref{E14}) as:
\begin{equation}\label{E18}
    C_{UB}(a)=C(a,x_{1, LB})
\end{equation}
\end{theorem}
\begin{proof}
For convenience, the proof is presented  in Appendix \ref{app-3}.
\end{proof}\

\section{Numerical Results and Discussion}\label{result}
The curves in Fig. \ref{Fig2} show respectively, the non-zero mass
point location of the capacity-achieving input distribution
$x_{1}$ obtained using maximization (\ref{E10}), and the one
obtained using relation (\ref{E13}) or equivalently (\ref{E15}).
As can be seen from Fig. \ref{Fig2}, the two curves are
undistinguishable at low SNR, confirming that (\ref{E16}) is exact
at low SNR. As the SNR increases, a small discrepancy between the
two curves starts to appear. This is expected since (\ref{E15})
holds for up to an order of magnitude strictly smaller than 2 and
thus for small SNR values, (but not smaller than about
$2.10^{-2}$), a discrepancy may appear. Nevertheless, even for an
SNR greater than $2.10^{-2}$, the curve obtained using (\ref{E15})
is very instructive especially as it follows the same shape as the
one obtained by simulation results. An interesting future work
would be to use (\ref{E16}) in order to understand why a new mass
point should appear as the SNR increases. It should be mentioned
that the discrepancy observed in Fig. \ref{Fig2} may be rendered
as small as desired using high order series expansion. However,
the analysis would be unrewardingly too complex.\

Figure \ref{Fig3} depicts the non-coherent capacity curves. Again,
the curve obtained by computer simulation and the one obtained
using (\ref{E14}) are undistinguishable. More interestingly, the
discrepancy observed at not very low SNR values in Fig. \ref{Fig2}
has vanished, implying that the capacity is not very sensitive to
the non-zero mass point location. Also shown in Fig. \ref{Fig3} is
the linear approximation $C(a)=a$, which is an upper bound on the
capacity. As can be noticed in Fig. \ref{Fig3}, the linear
approximation follows the same shape as the exact non-coherent
capacity curves at low SNR and becomes quite loose for SNR values
greater than $10^{-2}$. This implies that the sub-linear term
defined in (\ref{E16-subdef}) is much more important at these SNR
values. This can be seen in Fig. \ref{Fig4} where we have plotted
the non-coherence penalty percentage given by (\ref{E16-4}).
Figure \ref{Fig4} confirms that there is no substantial gain in
the channel knowledge in a capacity sense at very low SNR, thus
indicating that non-coherent communication is almost as
power-efficient as AWGN and coherent communications. As the SNR
increases, a non-coherence penalty begins to appear reaching up to
$70\%$.\

The derived upper and lower bounds on the non zero mass point
locations given respectively by (\ref{E16second}) and (\ref{E17})
as well as well as the bounds derived in \cite{Zheng-2007} are
plotted in Fig. \ref{Fig5} along with the exact curves at low SNR.
As can be seen in Fig. \ref{Fig5}, the upper bound in
\cite{Zheng-2007}, albeit tighter than (\ref{E16second}), crosses
the exact curves at about $2.10^{-2}$. At these not so low SNR
values, the derived bound in \cite{Zheng-2007} is no longer an
upper bound, consistently with our discussion in Subsection
\ref{subsectionA}. On the other hand, the lower bound (\ref{E17})
is tighter than the one derived in \cite{Zheng-2007} for all SNR
values.

\section{Conclusion}\label{conclusion}
In this paper, we have addressed the analysis of the capacity of
discrete-time non-coherent memoryless Rayleigh fading channels at
low SNR. We have computed explicitly the channel mutual
information at low SNR which is also a lower bound on the channel
mutual information, albeit not necessarily at low SNR values.\

Using the derived expression of the channel mutual information, we
have been able to provide a fundamental relation between the
non-zero mass point location of the capacity-achieving input
distribution and the SNR. This fundamental relation brings the
complete answer about how the optimal input distribution varies
with the power constraint at low SNR. It also provides an
analytical explanation on what was previously observed through
computer simulation in \cite{Abou-Faycal-01} about the peculiar
behavior of the non-zero mass point location at low SNR values.
The exact non-coherent capacity has been derived and insights on
the capacity behavior which can be gained through functional
analysis has been shown.\

In order to better understand how the non-zero mass point location
varies with the SNR, we have also derived lower and upper bounds
which have been compared to recently derived bounds. The newly
derived lower bound is tighter for all SNR values of interest,
whereas somewhat looser, the upper bound was shown to hold for
larger SNR values.

\appendices
\section{Proof of lemma \ref{L1}}\label{app-1}
For convenience, we will use $f(x)$ instead of $f_{x}(x)$ to
denote the probability density function of the random variable $x$
at the value $x$. We first prove that $I_{LB}(x;y)$ is a strictly
monotonically increasing function with respect to
$x_{1}$.\footnote{Note that the technic used here to prove that
$I_{LB}(x;y)$ is strictly monotonically increasing function with
respect to $x_{1}$ follows along the same lines as the technic
used to establish that the optimal input distribution has
necessarily a mass point at zero in \cite{Abou-Faycal-01}, albeit
the two technics have strictly different objectives}
Differentiating (\ref{E9}) with respect to $x_{1}$ yields
\renewcommand{\theequation}{\thesection.\arabic{equation}}
\begin{equation}\label{EA13}
    \frac{\partial}{\partial{x_{1}}}I_{LB}(x_{1},p_{1})=
p_{1}\int_{0}^{\infty}{\frac{\partial}{\partial{x_{1}}}}f(y|x_{1})\ln{\left(
\frac{f(y|x_{1})}{f(y)}\right)}dy
\end{equation}
Differentiating (\ref{E4}), we obtain:
\begin{equation}\label{EA14}
\frac{\partial}{\partial{x_{1}}}f(y|x_{1})=\frac{2x_{1}}{(1+x_{1}^{2})^2}\left[y-(1+x_{1}^{2})\right]f(y|x)
\end{equation}
Substituting (\ref{EA14}) in (\ref{EA13}) yields:
\begin{equation}\label{EA15}
    \frac{\partial}{\partial{x_{1}}}I_{LB}(x_{1},p_{1})=
    \frac{2p_{1} x_{1}}{(1+x_{1}^{2})^{2}}\int_{0}^{\infty}
\left[y-(1+x_{1}^{2})\right]f(y|x_{1})
    \ln{\left(\frac{f(y|x_{1})}{f(y)}\right)}dy
\end{equation}
Let $g(y)$ be defined as
$g(y)=\ln{\left(\frac{f(y|x_{1})}{f(y)}\right)}$. Now, we need the
following lemma.
\begin{Lemma}\label{L2}
Let $f(y)$ be a probability density function with mean $m$. If
$g(y)$ is a strictly monotonically increasing function then
\begin{equation}\label{EA16}
    \int{(y-m)f(y)g(y)} > 0
\end{equation}
\end{Lemma}
\begin{proof}
The proof follows along similar lines as Lemma 1 in
\cite{Abou-Faycal-01}.
\end{proof}
To apply Lemma \ref{L2}, it is sufficient to note that
\begin{equation}\label{EA17}
\frac{f(y)}{f(y|x_{1})}=p_{1}+p_{0}(1+x_{1}^2)\exp{\left[y\left(\frac{1
}{1+x_{1}^2}-1)\right)\right]}
\end{equation}
is strictly decreasing with respect to $y$ because the exponent of
the exponential function is negative, therefore
$\frac{f(y|x_{1})}{f(y)}$ is strictly increasing and so is $g(y)$.
Finally, using the fact that $(1+x_{1}^2)$ is the mean of
$f(y|x_{1})$ and applying Lemma \ref{L2} to (\ref{EA15}), we
obtain:
\begin{equation}\label{EA18}
    \frac{\partial}{\partial{x_{1}}}I_{LB}(x_{1},p_{1}) > 0,
\end{equation}
which means that $I_{LB}(x_{1},p_{1})$ is strictly increasing with
respect to $x_{1}$. Consequently, the average power constraint
holds with equality. That is $E[x^2]=p_{1}x_{1}^2=a$. Hence
(\ref{E8}) is equivalent to:
\begin{equation}\label{EA19}
     \begin{cases}
 C_{LB}=\underset{x_{1} \geq \sqrt{a}}{\max}{I_{LB}(x_{1},p_{1})}\\
p_{1}x_{1}^2=a.
\end{cases}
\end{equation}
Next, we prove the existence of the maximum in (\ref{EA19}).
Clearly, $I_{LB}(x_{1},p_{1})$ is now a function of $x_{1}$ and
$a$ since $p_{1}x_{1}^2=a$. That $x_{1} \geq \sqrt{a}$ follows
automatically from the fact that $p_{1}\leq1$. On the other hand,
$I_{LB}(x_{1},p_{1})$ in (\ref{E9}) is positive-definite and
continue with respect to $x_{1}$ and $p_{1}$ and thus so is
$I_{LB}(x_{1},a)$ for a given SNR value $a$. Moreover
$I_{LB}(x_{1},a)$ is upper-bounded over the interval
$[\sqrt{a},\infty[$ otherwise, one would have, for some SNR value,
say $a^{0}$ :
\begin{equation}\label{EA20}
    \forall \ \epsilon \ > 0, \quad \exists \ x_{1}^{0} \ > \ \sqrt{a^{0}}
\quad |
    \quad
    I_{LB}(x_{1}^{0},a^{0}) \ > \ \epsilon.
\end{equation}
But this statement also means that the channel mutual
information-an upper bound on $I_{LB}(x_{1},a^{0})$- is unbounded
for $a^{0}$ which contradicts the fact that the capacity exists
for all SNR values as proven in \cite{Abou-Faycal-01}. Hence,
$I_{LB}(x_{1},a)$ is necessarily upper-bounded. Furthermore, the
continuity of $I_{LB}(x_{1},a)$ over $[\sqrt{a},\infty[$ implies
that the upper-bound is either achieved at a finite value $x_{1}$
or at $\infty$. The last case is however impossible. To see this,
it is sufficient to observe that for a given $a$, as $x_{1}$ goes
to infinity, $p_{1}$ tends toward zero. Thus following (\ref{E9}),
$\underset{x_{1}\rightarrow\infty}{\lim}I_{LB}(x_{1},a)=I_{LB}(\infty,0)=0$,
and consequently $I_{LB}(x_{1},a)=0$ for all $x_{1} \in
[\sqrt{a},\infty[$ which is impossible since the discrete input
distribution $x$ and the output $y$ are dependent. That is, the
upper bound is achieved at a finite value $x_{1}$ and this proves
the existence of the maximum in (\ref{EA19}). Moreover, since the
maximum is not at the borders of $[\sqrt{a},\infty[$, we
necessarily have at the maximum $\frac{\partial}{\partial
x_{1}}I_{LB}(x_{1},a)=0$.

Finally, in order to prove (\ref{E11}), we directly compute the
lower bound $I_{LB}(x_{1},p_{1})$ from (\ref{E9}):
\begin{eqnarray}\label{EA21}
I_{LB}(x_{1},p_{1})&=&
\underbrace{p_{0}\int_{0}^{\infty}f(y|0)\ln{(f(y|0))}dy}_{I_{1}}-\underbrace
{
p_{0}\int_{0}^{\infty}f(y|0)\ln{(f(y))}}_{I_{2}} \nonumber \\
           & & +
\underbrace{p_{1}\int_{0}^{\infty}f(y|x_{1})\ln{(f(y|x_{1}))}}_{I_{3}}-
\underbrace{p_{1}\int_{0}^{\infty}f(y|x_{1})\ln{(f(y))}}_{I_{4}}
\end{eqnarray}
$I_{1}$ and $I_{3}$ may be easily computed:
\begin{equation}\label{EA22}
    I_{1}=p_{0}\int_{0}^{\infty}e^{-y}\ln{(e^{-y})}dy=-p_{0}=1-p_{1}
\end{equation}
\begin{eqnarray}\label{EA23}
I_{3}&=&p_{1}\int_{0}^{\infty}\frac{1}{1+x_{1}^2}e^{-\frac{y}{1+x_{1}^2}}\ln
{\left(\frac{1}{1+x_{1}^2}e^{-\frac{y}{1+x_{1}^2}}\right)}dy \nonumber \\
    &=&-p_{1}\left(1+\ln{(1+x_{1}^2)}\right)
\end{eqnarray}
\begin{eqnarray}\label{EA24}
I_{2}&=&p_{0}\int_{0}^{\infty}e^{-y}\ln{\left(p_{0}e^{-y}+\frac{p_{1}}{1+x_{
1}^2}e^{-\frac{y}{1+x_{1}^2}}\right)}dy \nonumber \\
&=&\underbrace{\int_{0}^{\infty}p_{0}e^{-y}\ln{\left(p_{0}e^{-y}\right)}dy}_
{I_{21}}+\underbrace{\int_{0}^{\infty}p_{0}e^{-y}\ln{\left(1+\frac{p_{1}}{p_
{0}(1+x_{1}^2)}e^{\left(1-\frac{1}{1+x_{1}^2}\right)y}\right)}dy}_{I_{22}}
\end{eqnarray}
$I_{21}$ can be easily computed:
\begin{equation}\label{EA25}
    I_{21}=p_{0}\left[\ln{(p_{0})}-1\right]
\end{equation}
In order to compute $I_{22}$, let $\alpha=1+x_{1}^2$ and
$\beta=\frac{p_{1}}{p_{0}\alpha}=\frac{p_{1}}{(1-p_{1})\alpha}$.
Thus, $I_{22}$ may be written:
\begin{eqnarray}\label{EA26}
I_{22}&=&
\frac{p_{0}\alpha}{\alpha-1}\int_{1}^{\infty}t^{\frac{1-2\alpha}{\alpha-1}\ln{(1+\beta t)}}dt \nonumber\\
      &=&
\frac{p_{0}\alpha}{\alpha-1}\left\{\left[\frac{1-\alpha}{\alpha}t^{-\frac{\alpha}{\alpha-1}}\ln{(1+\beta
t)}\right]_{1}^{\infty}-\frac{1-\alpha}{\alpha}\beta
\int_{1}^{\infty}\frac{t^{-\frac{\alpha}{\alpha-1}}}{1+\beta
t}dt\right\}
\end{eqnarray}
The integral on the RHS of (\ref{EA26}) may be computed as
\cite{Ryzhik}:
\begin{equation}\label{EA27}
    \int_{1}^{\infty}\frac{t^{-\frac{\alpha}{\alpha-1}}}{1+\beta
    t}dt=\frac{\alpha-1}{\alpha \beta}\cdot
{}_{2}F_{1}\left(1,1+\frac{1}{\alpha -1},2+\frac{1}{\alpha
-1},-\frac{1}{\beta} \right)
\end{equation}
Substituting (\ref{EA27}) in (\ref{EA26}), we obtain:
\begin{equation}\label{EA28}
    I_{22}=p_{0}\left[\ln{(1+\beta)}+\frac{\alpha
-1}{\alpha}{}_{2}F_{1}\left(1,1+\frac{1}{\alpha
-1},2+\frac{1}{\alpha -1},-\frac{1}{\beta}
    \right)\right],
\end{equation}
and thus combining (\ref{EA24}), (\ref{EA25}) and (\ref{EA28}),
yields:
\begin{equation}\label{EA29}
    I_{2}= p_{0}\left[\ln{(p_{0})}-1\right]+
p_{0}\left[\ln{(1+\beta)}+\frac{\alpha
-1}{\alpha}\cdot{}_{2}F_{1}\left(1,1+\frac{1}{\alpha
-1},2+\frac{1}{\alpha -1},-\frac{1}{\beta}
    \right)\right].
\end{equation}
The integral $I_{4}$ may be computed similarly. We skip the
details and give below the final result:
\begin{equation}\label{EA30}
    I_{4}= p_{1}\ln{(p_{0})}-p_{1}\alpha+ p_{1}\left[\ln{(1+\beta)}+(\alpha
-1)\cdot{}_{2}F_{1}\left(1,\frac{1}{\alpha -1},1+\frac{1}{\alpha
-1},-\frac{1}{\beta}
    \right)\right].
\end{equation}
Following (\ref{EA21}), (\ref{EA22}), (\ref{EA23}), (\ref{EA29}),
(\ref{EA30}) and using the fact that:
\begin{equation}\label{EA31}
    {}_{2}F_{1}\left(1,\frac{1}{\alpha -1},1+\frac{1}{\alpha
-1},-\frac{1}{\beta}
    \right)+\frac{1-p_{1}}{p_{1}}\cdot {}_{2}F_{1}\left(1,1+\frac{1}{\alpha
-1},2+\frac{1}{\alpha
    -1},-\frac{1}{\beta}
    \right)=1,
\end{equation}
we obtain:
\begin{eqnarray}\label{EA32}
I_{LB}(x_{1},p_{1})&=&-\ln{(1-p_{1})}+p_{1}\left(x_{1}^{2}-\ln{(1+x_{1}^{2})
}\right)-\ln{(1+\beta)} \nonumber \\
    & & -\frac{p_{1}(\alpha -1)}{\alpha}\left[(\alpha -1)\cdot
{}_{2}F_{1}\left(1,\frac{1}{\alpha -1},1+\frac{1}{\alpha
-1},-\frac{1}{\beta}
    \right) + 1 \right].
\end{eqnarray}
Combining (\ref{EA32}) and (\ref{EA19}) yields (\ref{E11}) which
completes the proof of Lemma \ref{L1}.
\section{Proof of theorem  \ref{T1}}\label{app-2}
At low SNR, a discrete input distribution with two mass points,
one of them located at zero, achieves the non-coherent capacity
\cite{Abou-Faycal-01}.  That $p_{1}=a/x_{1}^{2}$ was proven in
Appendix \ref{app-1}. Therefore, (\ref{E12}) is true. To derive
(\ref{E13}), it is a matter of series expansion calculus.\

Before proceeding, it should be reminded that for the optimal
input distribution given in Theorem \ref{T1}, the non-zero mass
point location $x_{1}$ is greater than 1 $(x_{1} > 1)$
\cite{Abou-Faycal-01,Zheng-2007}. Then, series expansion of
(\ref{E11})  to the second order, around the point
$(x_{1},a)=(x_{1},0)$, where $x_{1}$ is an arbitrary real greater
than one, can be obtained using Mathematica:
\begin{eqnarray}\label{EA33}
I_{LB}(x_{1},a)&=&\biggl((1-\frac{\log(1+x_{1}^2)}{x_{1}^2})a+\frac{1}{2(x_{
1}^{2}-1)}a^{2}+\textit{o}{(a^{2})}\biggr) \nonumber \\
     & & -a^{\frac{1}{x_{1}^2}} \biggl(\pi
x_{1}^2\biggl(x_{1}^2(1+x_{1}^2)\biggr)^{-\frac{1+x_{1}^2}{x_{1}^2}}\csc{\biggl(\frac{\pi}{x_{1}^2}\biggr)} a \nonumber \\
    & & + \frac{\pi
\biggl(x_{1}^2(1+x_{1}^2)\biggr)^{-\frac{1+x_{1}^2}{x_{1}^2}}\csc{\biggl(\frac{\pi}{x_{1}^2}}\biggr)}{x_{1}^2}
    a^{2}+\textit{o}{(a^{2})}\biggr),
\end{eqnarray}
where the symbol $\circ{(a^{n})}$ represents a function say
$g(x_{1},a)$, such that $\underset{a \rightarrow
0}{\lim}\frac{g(x_{1},a)}{a^n}=0$. Since $x_{1} > 1$, then there
exists $\epsilon
> 0$ such that $1+\frac{1}{x_{1}^2} < 2-\epsilon$. Thus,
(\ref{EA13}) may be written as:
\begin{equation}\label{EA34}
    I_{LB}(x_{1},a)=\biggl(1-\frac{\log(1+x_{1}^2)}{x_{1}^2}\biggr)a
    - \pi
x_{1}^2\biggl(x_{1}^2(1+x_{1}^2)\biggr)^{-\frac{1+x_{1}^2}{x_{1}^2}}\csc{\biggl(\frac{\pi}{x_{1}^2}\biggr)}a^{1+\frac{1}{x_{1}^2}}+\textit{o}{(a^{2-\epsilon})},
\end{equation}
which represents series expansion to an order strictly less than
2. Up to this order, we may make some abuse of notation, drop the
term $\textit{o}{(a^{2-\epsilon})}$ and write (\ref{EA34}) as:
\begin{equation}\label{EA35}
    I_{LB}(x_{1},a)=\biggl(1-\frac{\log(1+x_{1}^2)}{x_{1}^2}\biggr)a
    - \pi
x_{1}^2\biggl(x_{1}^2(1+x_{1}^2)\biggr)^{-\frac{1+x_{1}^2}{x_{1}^2}}\csc{\biggl(\frac{\pi}{x_{1}^2}\biggr)}a^{1+\frac{1}{x_{1}^2}}.
\end{equation}
Maximizing (\ref{EA15}) with respect to $x_{1}>1$ is equivalent
to:
\begin{equation}\label{EA36}
    \min_{x_{1}>1}\Biggl[ \frac{\log(1+x_{1}^2)}{x_{1}^2} + \pi
x_{1}^2\biggl(x_{1}^2(1+x_{1}^2)\biggr)^{-\frac{1+x_{1}^2}{x_{1}^2}}\csc{\biggl(\frac{\pi}{x_{1}^2}\biggr)}a^{1+\frac{1}{x_{1}^2}}
    \Biggr].
\end{equation}
As was proven in Appendix \ref{app-1}, at the maximum, we have
necessarily $\frac{\partial}{\partial x_{1}}I_{LB}(x_{1},a)=0$.
Differentiating (\ref{EA36}) with respect to $x_{1}$ yields
(\ref{E13}). Finally, (\ref{E14}) follows from (\ref{EA35}). This
completes the proof of Theorem \ref{T1}.
\section{Proof of theorem  \ref{T2}}\label{app-3}
For $a<a_{0}$ and $x_{1} > x_{0}$, (\ref{E15}) may be written as:
\begin{equation}\label{EA37}
    a(x_{1})=
\exp{\bigl[x_{1}^{2}W\bigl(-1,\varphi{(x_{1})}\bigr)-x_{1}^{2}+\pi
\cot{\bigl(\frac{\pi}{x_{1}^{2}}\bigr)}+\ln{(x_{1}^{2})}+\ln{(1+x_{1}^{2})}-
1
    \bigr]}.
\end{equation}
Moreover, it is easy to check that $a$ in (\ref{EA37}) is a
decreasing function with respect to $x_{1}$ and that:
\begin{equation}\label{EA38}
-x_{1}^{2}+\pi
\cot{\bigl(\frac{\pi}{x_{1}^{2}}\bigr)}+\ln{(x_{1}^{2})}+\ln{(1+x_{1}^{2})}-
1
> 1,
\end{equation}
for $x_{1} > x_{0}$. Thus, using (\ref{EA37}) and (\ref{EA38}), we
have:
\begin{equation}\label{EA39}
    a(x_{1}) >
a_{lb}(x_{1})=\exp{\bigl[x_{1}^{2}W\bigl(-1,\varphi{(x_{1})}\bigr)+1\bigr]},
\end{equation}
where $a_{lb}(x_{1})$ is a lower bound on $a(x_{1})$. Since
$a_{lb}(x_{1})$ is also a decreasing function with respect to
$x_{1}$, then for a low SNR value $a$, (\ref{EA39}) may be seen as
a lower bound on the optimal non-zero mass point location $x_{1}$
and we equivalently have:
\begin{equation}\label{EA40}
    x_{1} > x_{1,lb},
\end{equation}
where $x_{1,lb}$ is the solution of $a_{lb}(x_{1})=a$. Next, we
derive a lower bound on $x_{1,lb}$.\

Let us fixe a low SNR value $a < a_{0}$ and consider the function
on the RHS of (\ref{EA39}) written for simplicity as:
\begin{equation}\label{EA41}
    a=\exp{\bigl[x_{1,lb}^{2}W\bigl(-1,\varphi{(x_{1,lb})}\bigr)+1\bigr]},
\end{equation}
or equivalently by letting
$y=\sqrt{1+\ln{\bigl(\frac{1}{a}\bigr)}}$:
\begin{equation}\label{EA42}
    x_{1,lb}^2=\frac{y^2}{-W\bigl(-1,\varphi{(x_{1,lb})}\bigr)}.
\end{equation}
Since $-W\bigl(-1,\varphi{(x_{1,lb})}\bigr) > 1$ for
$x_{1,lb}>x_{0}$, it is easy to see that $y^2 > x_{1,lb}^2$.
Hence, using the fact that $\varphi{(\cdot)}$ and
$-W\bigl(-1,\cdot\bigr)$ are strictly increasing functions, we
have:
\begin{equation}\label{EA43}
    x_{1,LB}^{(1)}=\frac{y}{\sqrt{-W\bigl(-1,\varphi{(y)}\bigr)}} <
    x_{1,lb}=\frac{y}{\sqrt{-W\bigl(-1,\varphi{(x_{1,lb})}\bigr)}},
\end{equation}
where the superscript $(^{1})$ on the left hand side of
(\ref{EA43}) means a first lower bound. Next we improve the lower
bound $x_{1,LB}^{(1)}$ to obtain a tighter one. But before going
on, we remind this result from \cite{Pickett} which aims at
resolving transcendental equations involving Lambert function
iteratively using self-mapping techniques:
\begin{Lemma}\label{L4}
For the region specified by $x < 1$ and $-\frac{1}{e} < y < 0$, an
infinite-ladder solution to the equation:
\begin{equation}\label{EA44}
    y(x)=x e^{x}
\end{equation}
is easily identified as
\begin{equation}\label{EA45}
    x(y)=L_{<}(y),
\end{equation}
with the ladder $L_{<}(y)$ defined as
\begin{equation}\label{EA46}
    L_{<}(y)=-\ln{\Biggl( \frac{\ln{\frac{\ln{(\ldots)}}{-y}}}{-y}
    \Biggr)}.
\end{equation}
\end{Lemma}
\begin{proof}
The proof and more details concerning the Lambert function can be
found in \cite{Pickett}.
\end{proof}
Clearly, using (\ref{EA46}) and the fact that the solution of
(\ref{EA44}) is also $x(y)=W(-1,y)$, one can obtain a simple upper
bound on the Lambert function in the interval of interest:
\begin{equation}\label{EA47}
    W(-1,y) \leq \ln{(-y)}-\ln{\bigl(-\ln{(-y)}\bigr)}.
\end{equation}
Since for $x_{1,lb}>x_{0}$, $\varphi{(x_{1,lb})} \in
]-\frac{1}{e},0[$ and $W(-1,\varphi{(x_{1,lb})}) < 0$, then
applying (\ref{EA47}) to $\varphi{(x_{1,lb})}$ yields:
\begin{eqnarray}
\label{EA48}
  W(-1,\varphi{(x_{1,lb})})&\leq&
\ln{\bigl(\frac{-\varphi{(x_{1,lb})}}{-\ln{\bigl(-\varphi{(x_{1,lb})}
\bigr)}} \bigr)} \\
  \label{EA49}
   &\leq & \ln{\bigl(\frac{-\varphi{(y)}}{-\ln{\bigl(-\varphi{(x_{1,lb})}
\bigr)}} \bigr)} \\
   \label{EA50}
   &\leq& \ln{\bigl(-\varphi{(y)} \bigr)}.
\end{eqnarray}
Inequality (\ref{EA49}) holds because $y>x_{1,lb}$ and
$\varphi{(\cdot)}$ is an increasing function, likewise
(\ref{EA50}) follows from the fact that for $x>x_{0}$,
$\varphi{(x)}> -\frac{1}{e}$ and thus
$\frac{1}{-\ln{\bigl(-\varphi{(x)}\bigr)}} < 1$. Moreover,
(\ref{EA50}) implies
\begin{equation}\label{EA51}
    \frac{y}{-\ln{\bigl(-\varphi{(y)}\bigr)}} \geq
    \frac{y}{-W(-1,\varphi{(x_{1,lb})})}=x_{1,lb}
\end{equation}
Applying again respectively $\varphi{(\cdot)}$ and $-W(-1,\cdot)$
to both sides of (\ref{EA51}) gives:
\begin{equation}\label{EA52}
x_{1,LB}^{(2)}=\frac{y}{\sqrt{-W\biggl(-1,\varphi{\Bigl(\frac{y}{-\ln{\bigl(
-\varphi{(y)}\bigr)}}\Bigr)}\biggr)}} \leq
   \frac{y}{\sqrt{-W(-1,\varphi{(x_{1,lb})})}}=x_{1,lb}.
\end{equation}
Finally, to prove that $x_{1,LB}^{(2)}$ is tighter than
$x_{1,LB}^{(1)}$, it is sufficient to note that since
$\varphi{(x_{1,lb})} \in ]-\frac{1}{e},0[$, $y > x_{1,lb}$ and
$\varphi{(\cdot)}$ is an increasing function, then $\varphi{(y)}
\in ]-\frac{1}{e},0[$ and we have consequently: $y
> \frac{y}{-\ln{\bigl(-\varphi{(y)} \bigr)}}$. Applying again respectively
$\varphi{(\cdot)}$ and $-W(-1,\cdot)$ to this inequality yields:
\begin{equation}\label{EA53}
    x_{1,LB}^{(1)}=\frac{y}{\sqrt{-W\bigl(-1,\varphi{(y)}\bigr)}} \leq
\frac{y}{\sqrt{-W\biggl(-1,\varphi{\Bigl(\frac{y}{-\ln{\bigl(-\varphi{(y)}\bigr)}}\Bigr)}\biggr)}}=x_{1,LB}^{(2)}.
\end{equation}
Combining (\ref{EA52}) and (\ref{EA53}), we have:
\begin{equation}\label{EA54}
    x_{1,LB}^{(1)} \leq x_{1,LB}^{(2)} \leq x_{1,lb},
\end{equation}
from which (\ref{E17}) follows by letting
$x_{1,LB}^{(2)}=x_{1,LB}$. Finally, (\ref{E18}) may be obtained by
applying (\ref{E14}) to $x_{1,LB}$. This completes the proof of
Theorem \ref{T2}.

\newpage
\begin{figure*}
\centerline{\subfigure[Very Low
SNR]{\includegraphics[width=2.5in]{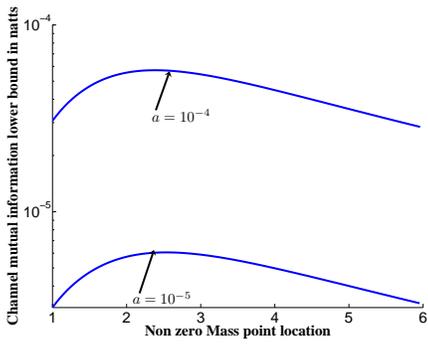}\label{Fig11}} \hfil
\subfigure[Low
SNR]{\includegraphics[width=2.5in]{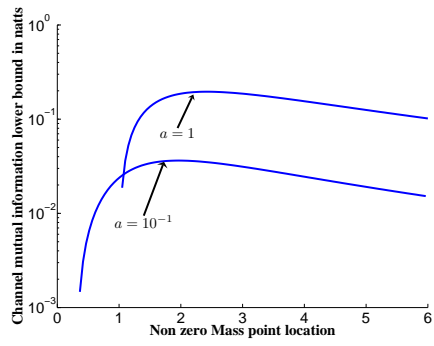}\label{Fig12}} \hfil
\subfigure[High
SNR]{\includegraphics[width=2.5in]{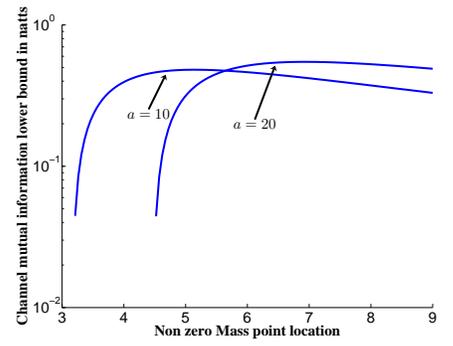}\label{Fig13}}}
\caption{Channel mutual information lower bound versus non-zero
mass point for 3 SNR regimes: a) Very Low SNR, b) Low SNR and c)
High SNR} \label{Fig1}
\end{figure*}

\clearpage

\begin{figure}[t]
  \begin{center}
    \includegraphics[scale=0.5]{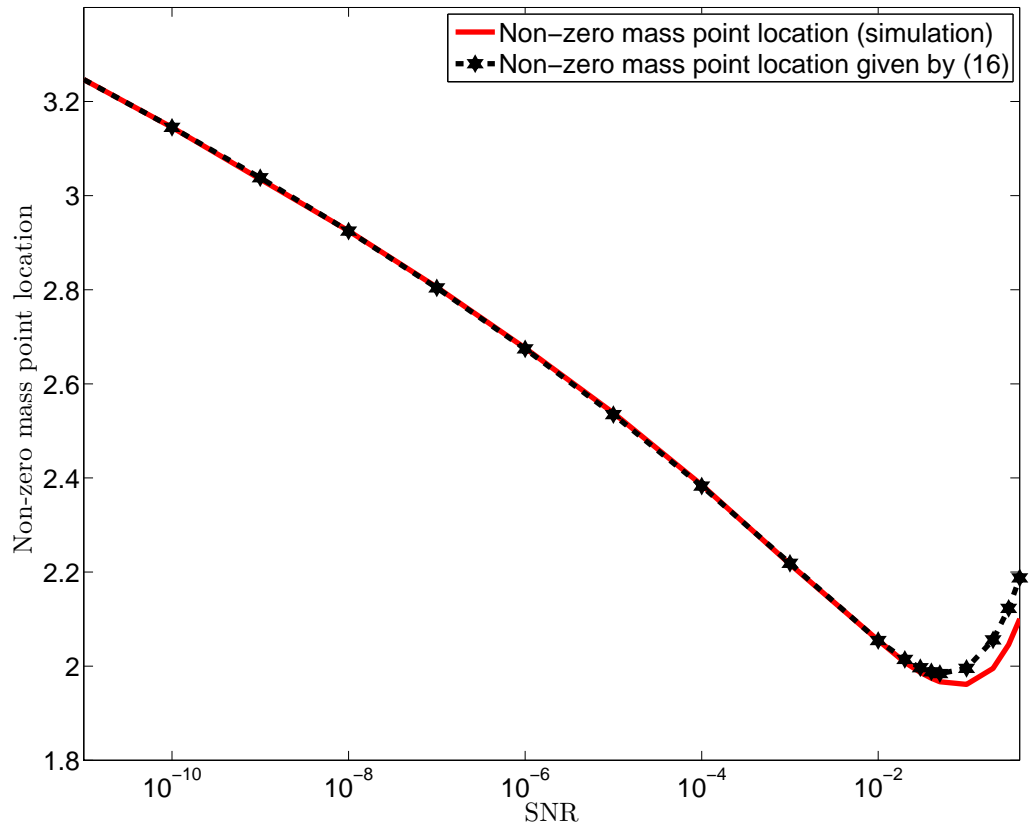}
        \caption{Location of non-zero mass point versus the SNR
        value $a$ (linear).}
    \label{Fig2}
  \end{center}
\end{figure}

\clearpage

\begin{figure}[t]
  \begin{center}
    \includegraphics[scale=0.5]{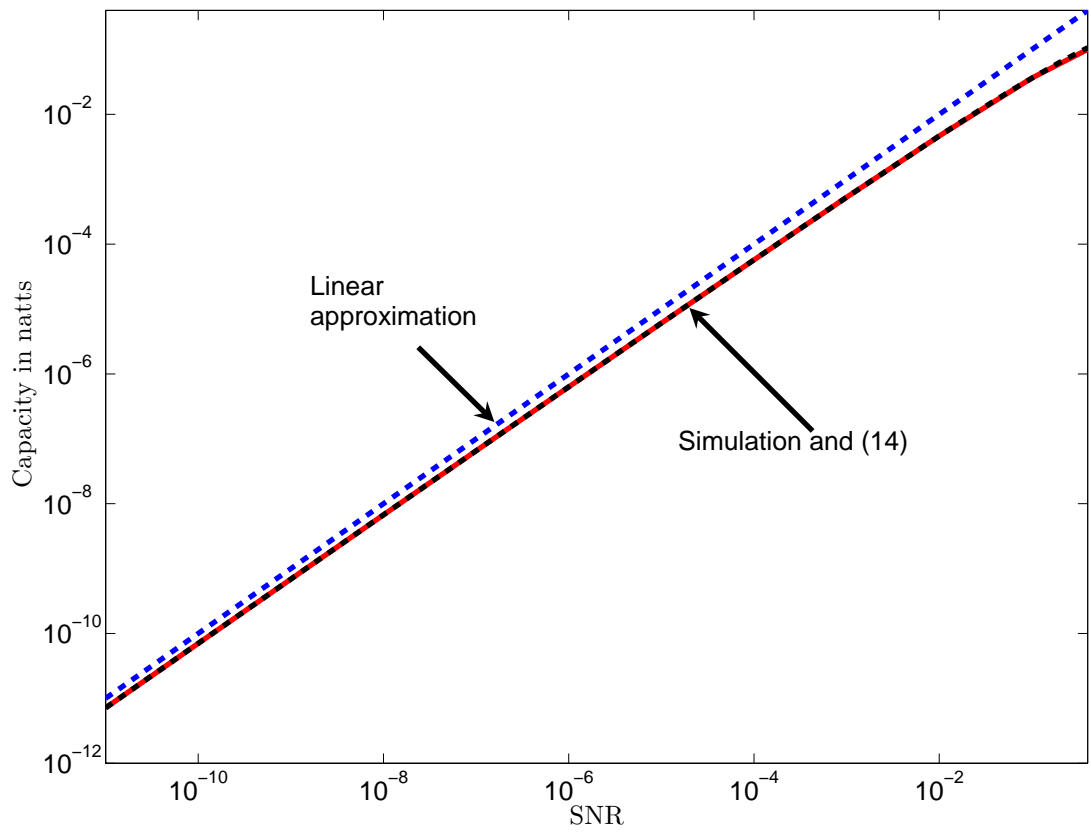}
        \caption{Non-coherent capacity versus the SNR
        value $a$ (linear).}
    \label{Fig3}
  \end{center}
\end{figure}

\clearpage

\begin{figure}[t]
  \begin{center}
    \includegraphics[scale=0.5]{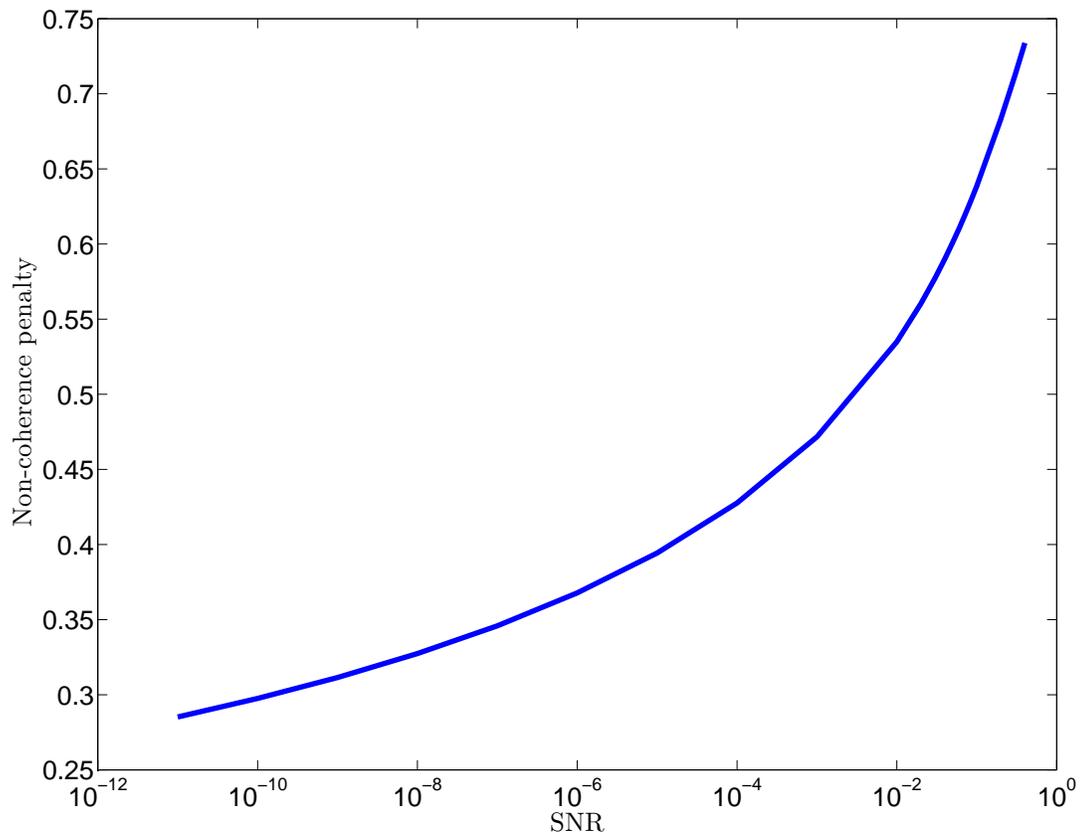}
        \caption{Non-coherentce penalty per SNR versus the SNR
        value $a$ (linear).}
    \label{Fig4}
  \end{center}
\end{figure}

\clearpage


\clearpage

\begin{figure}[t]
  \begin{center}
    \includegraphics[scale=0.5]{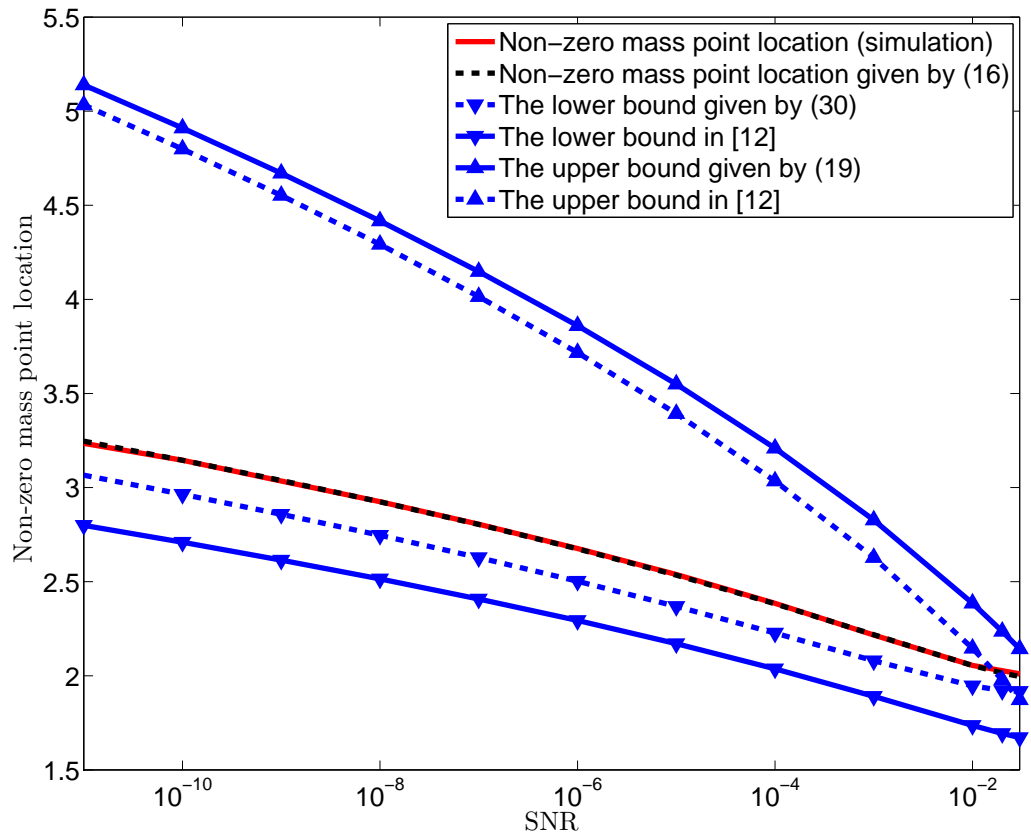}
        \caption{Exact non-zero mass point locations and the derived upper and lower bounds as well as those reported in \cite{Zheng-2007} versus the SNR
        value $a$ (linear).}
    \label{Fig5}
  \end{center}
\end{figure}

\end{document}